\documentclass{acm_proc_article-sp}
\usepackage[ruled, noend, noline, linesnumbered]{algorithm2e}
\usepackage{caption}
\usepackage{subcaption}
\usepackage[pagebackref,letterpaper=true,colorlinks=true,pdfpagemode=none,urlcolor=blue,linkcolor=blue,citecolor=violet,pdfstartview=FitH]{hyperref}
\usepackage{paralist}
\usepackage{graphicx}
\usepackage{color}
\usepackage{xcolor}
\usepackage{epstopdf}
\newtheorem{theorem}{Theorem}[section]

\newtheorem{lemma}[theorem]{Lemma}
\newtheorem{claim}[theorem]{Claim}
\newtheorem{corollary}[theorem]{Corollary}
\newtheorem{definition}[theorem]{Definition}

\newcommand{\ignore}[1]{}

\newcommand{\cS}{\mathcal{S}}

\newcommand{\eps}{\varepsilon}

\newcommand{\ceil}[1]{\lceil#1\rceil}

\newcommand{\EX}{\hbox{\bf E}}

\newcommand{\Sec}[1]{\hyperref[sec:#1]{\S\ref*{sec:#1}}} %section
\newcommand{\Eqn}[1]{\hyperref[eq:#1]{(\ref*{eq:#1})}} %equation
\newcommand{\Fig}[1]{\hyperref[fig:#1]{Fig.\,\ref*{fig:#1}}} %figure
\newcommand{\Tab}[1]{\hyperref[tab:#1]{Tab.\,\ref*{tab:#1}}} %table
\newcommand{\Thm}[1]{\hyperref[thm:#1]{Theorem\,\ref*{thm:#1}}} %theorem
\newcommand{\Fact}[1]{\hyperref[fact:#1]{Fact\,\ref*{fact:#1}}} %fact
\newcommand{\Lem}[1]{\hyperref[lem:#1]{Lemma\,\ref*{lem:#1}}} %lemma
\newcommand{\Prop}[1]{\hyperref[prop:#1]{Prop.~\ref*{prop:#1}}} %property
\newcommand{\Cor}[1]{\hyperref[cor:#1]{Corollary~\ref*{cor:#1}}} %corollary
\newcommand{\Conj}[1]{\hyperref[conj:#1]{Conjecture~\ref*{conj:#1}}} %conjecture
\newcommand{\Def}[1]{\hyperref[def:#1]{Definition~\ref*{def:#1}}} %definition
\newcommand{\Alg}[1]{\hyperref[alg:#1]{Alg.~\ref*{alg:#1}}} %algorithm
\newcommand{\Ex}[1]{\hyperref[ex:#1]{Ex.~\ref*{ex:#1}}} %example
\newcommand{\Clm}[1]{\hyperref[clm:#1]{Claim~\ref*{clm:#1}}} %example
\newcommand{\Obs}[1]{\hyperref[obs:#1]{Obs.~\ref*{obs:#1}}} %observation

\newcommand{\ind}{C}
\newcommand{\nonind}{N}

\newcommand{\samp}{k}

\newcommand{\sample}{{\tt sample}}
\newcommand{\samplecent}{{\tt sample-centered}}
\newcommand{\estimate}{{\tt 3-path-sampler}}
\newcommand{\estcent}{{\tt centered-sampler}}

\newcommand{\nice}{N} % non-induced count
\newcommand{\ice}{C} % induced  count
\newcommand{\aice}{\widehat{\ice}} % approximate induced count
\begin{document}

\title{Path Sampling: A Fast and Provable Method for Estimating 4-Vertex Subgraph Counts\titlenote{This work was funded by the GRAPHS Program at DARPA and by the Applied Mathematics Program at the U.S.\@ Department of Energy. Sandia National Laboratories is a multi-program laboratory managed and operated by Sandia Corporation, a wholly owned subsidiary of Lockheed Martin Corporation, for the U.S. Department of Energy's National Nuclear Security Administration under contract DE-AC04-94AL85000.}}
\numberofauthors{3} 
\author{
\alignauthor
Madhav Jha\\
       \affaddr{Sandia National Laboratories}\\
       \affaddr{Livermore, CA 94550}\\
       \email{mjha@sandia.gov}
% 2nd. author
\alignauthor
C. Seshadhri\\
       \affaddr{Sandia National Laboratories}\\
       \affaddr{Livermore, CA 94550}\\
       \email{scomand@sandia.gov}
% 3rd. author
\alignauthor
Ali Pinar\\
       \affaddr{Sandia National Laboratories}\\
       \affaddr{Livermore, CA 94550}\\
       \email{apinar@sandia.gov}
}

\newcommand{\Sesh}[1]{{\color{red} Sesh says: #1}}

% make the title area
\maketitle
\begin{abstract}
Counting the frequency of small subgraphs is a fundamental technique in network analysis across various domains, most notably in bioinformatics and social networks. The special case of triangle counting has received much attention. Getting results for 4-vertex patterns is highly challenging, and there are few practical results known that can scale to massive sizes. Indeed, even a highly tuned enumeration code takes more than a day on a graph with millions of edges. 
Most previous work that runs for truly massive graphs employ clusters and massive parallelization.

We provide a sampling algorithm that provably and accurately approximates the frequencies of all 4-vertex pattern subgraphs. Our algorithm is based on a novel technique of \emph{3-path sampling} and a special pruning scheme to decrease the variance in estimates. We provide theoretical proofs for the accuracy of our algorithm, and give formal bounds for the error and confidence of our estimates. We perform a detailed empirical study and show that our algorithm provides estimates within 1\% relative error for all subpatterns (over a large class of test graphs), while being orders of magnitude faster than  enumeration  and other sampling based algorithms. 
Our algorithm takes less than a minute (on a single commodity machine) to process an Orkut social network with 300 million edges.
\end{abstract}

\section{Introduction}

Counting the number of occurrences of small subgraphs in a graph is a fundamental network analysis technique used across diverse domains: bioinformatics, social sciences, and 
infrastructure networks studies~\cite{HoLe70,Co88,Po98,Milo2002,Burt04,PrzuljCJ04,Fa07,HoBe+07,BeBoCaGi08,Fa10,SzTh10,FoDeCo10,SonKanKim12}. The subgraphs whose counts are desired are variously referred as ``pattern subgraphs,'' ``motifs,'' or ``graphlets.'' It is repeatedly observed that certain small subgraphs occur substantially more often in real-world networks than in a randomly generated network~\cite{HoLe70,WaSt98,Milo2002}. Motifs distributions have been used in bioinformatics to evaluate network models~\cite{PrzuljCJ04,HoBe+07}. Analysis of triadic (3-vertex) motifs has a long
history in social network analysis and modeling~\cite{HoLe70,Burt04,Fa07,SeKoPi11,DuPiKo12}.
Work in the data mining community has applied motif frequencies for spam detection and
group classification of sets of nodes~\cite{BeBoCaGi08,UganderBK13}.

The main challenge of motif counting is combinatorial explosion. Even in a moderately sized graph with
millions of edges, the subgraph counts (even for 4-vertex patterns) is in the billions.
Any exhaustive enumeration method (no matter how cleverly designed) is forced to touch each occurrence
of the subgraph, and cannot truly scale. One may apply massive parallelism to counteract this problem, but that does
not avoid the fundamental combinatorial explosion. An alternative approach is based on \emph{sampling}. Here,
we try to count the number of subgraphs using a randomized algorithm. The difficulty is in designing
a fast algorithm that also provides an accurate estimate. The holy grail is to get
mathematically provable bounds on accuracy with quantifiable error bars.

Sampling approaches have been employed for triangle counting with good success~\cite{ScWa05-2,TsDrMi09,TsKaMiFa09,TsKoMi11,SePiKo13}. There also exists work
for counting larger motifs, as we shall discuss later. Most methods (especially in bioinformatics)~\cite{HoBe+07,BetzlerBFKN11,WongBQH12,RaBhHa14}
work for graphs of at most 100K edges, much smaller than
the massive social networks we encounter.

\begin{figure}[t]
    \centering
    \begin{subfigure}[b]{0.15\textwidth}
        \centering
        \includegraphics{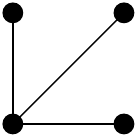}
        \caption{3-star}
        \label{fig:p2}
    \end{subfigure}
    \hfill
    \begin{subfigure}[b]{0.15\textwidth}
        \centering
        \includegraphics{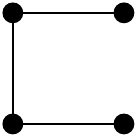}
        \caption{3-path}
        \label{fig:p1}
    \end{subfigure}
    \hfill
    \begin{subfigure}[b]{0.15\textwidth}
        \centering
        \includegraphics{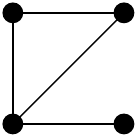}
        \caption{tailed-triangle}
        \label{fig:p3}
    \end{subfigure}
      \begin{subfigure}[b]{0.15\textwidth}
        \centering
        \includegraphics{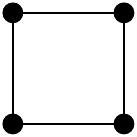}
        \caption{4-cycle}
        \label{fig:p4}
    \end{subfigure}
    \hfill
    \begin{subfigure}[b]{0.15\textwidth}
        \centering
        \includegraphics{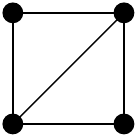}
        \caption{\footnotesize{chordal-4-cycle}}
        \label{fig:p5}
    \end{subfigure}
    \hfill
    \begin{subfigure}[b]{0.15\textwidth}
        \centering
        \includegraphics{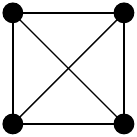}
        \caption{4-clique}
        \label{fig:p6}
    \end{subfigure}
    \caption{List of all connected 4-vertex motifs}
    \label{fig:4-vertex-motifs}
\end{figure}

\begin{figure*}[t]
		\centering
    \begin{subfigure}[b]{0.45\textwidth}
        \centering
    	\includegraphics[width=\textwidth]{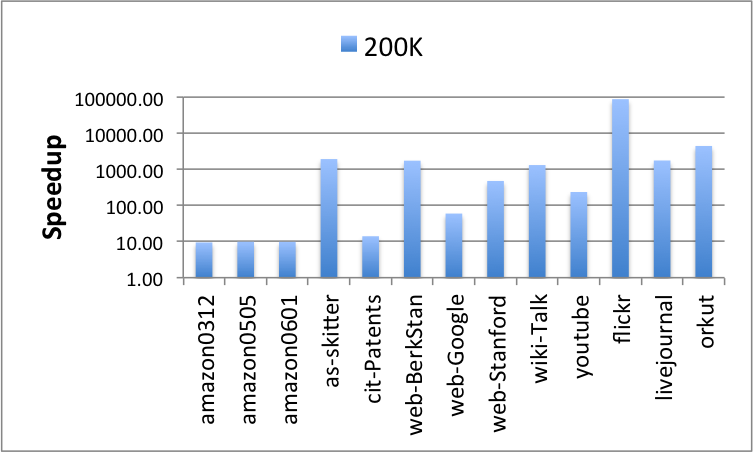}
   	\caption{Speedup}
    \label{fig:speedup}
    \end{subfigure}
   	\hfill
    \begin{subfigure}[b]{0.45\textwidth}
        \centering
    	\includegraphics[width=\textwidth]{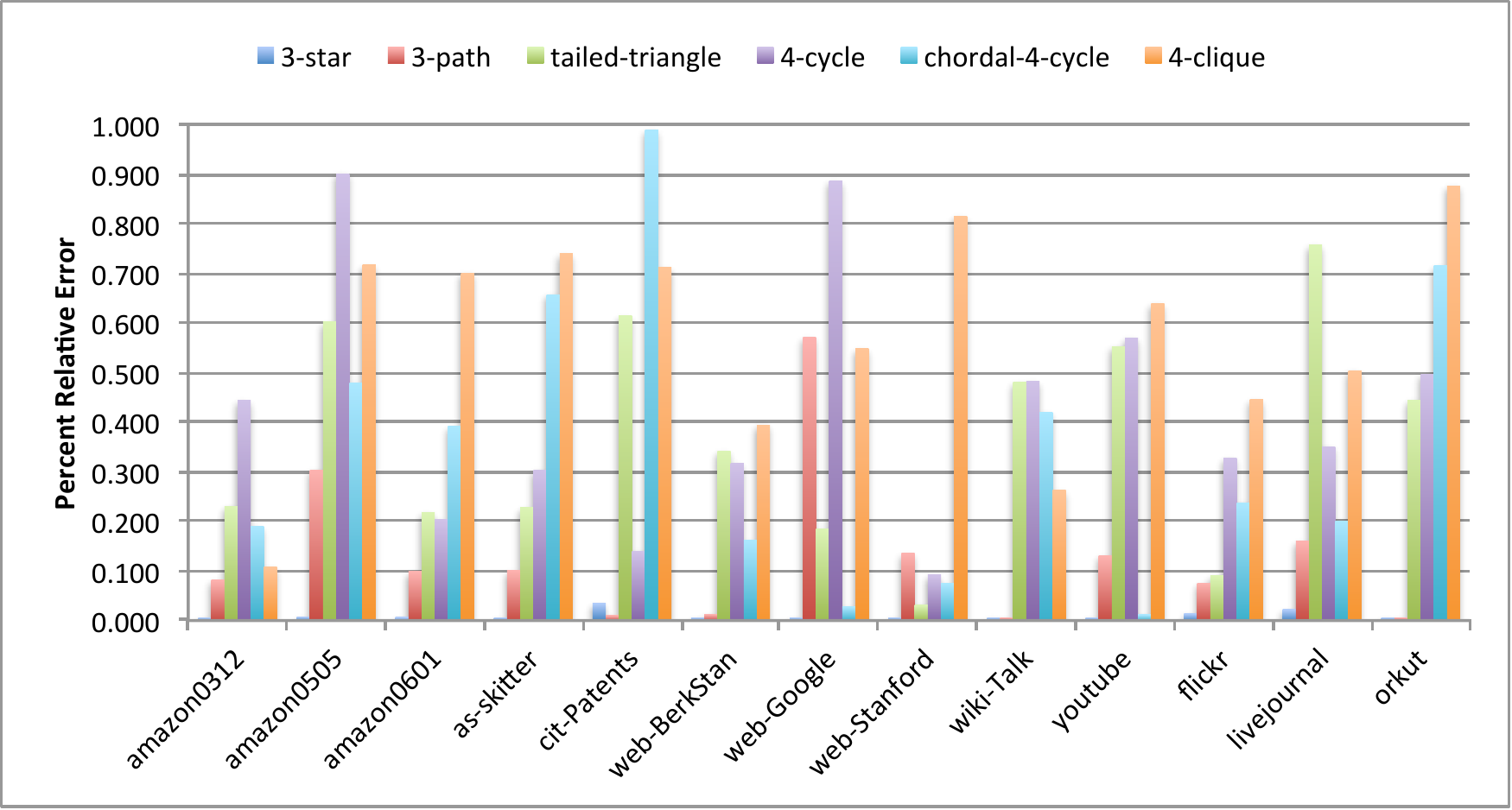}
   	\caption{Relative error}
    \label{fig:error}
    \end{subfigure}
    
  \caption{Summary of 3-path sampling algorithm behavior over a large variety of datasets: The left figure shows speedup over a tuned enumeration code. The right figure
   shows the relative error of each estimate, which is always less than 1\% (and mostly much smaller).}
\end{figure*}

\subsection{The main problem}

We focus on estimating frequency of all connected 4-vertex subgraphs on massive input graphs. There are six connected 4-vertex graphs (\Fig{4-vertex-motifs}): (i)  the 3-path, (ii) the 3-star, (iii) the tailed-triangle, (iv) the 4-cycle, (v) the chordal-4-cycle, and (vi) the 4-clique. Throughout this work, we refer to these motifs by their numbering in this list. For example the ``6-th motif'' is the 4-clique.

\emph{Our aim is to give an accurate
and fast estimate of all 4-vertex subgraph counts.}
Triadic analysis is now a standard aspect of network analysis. 
Recent work of Ugander et al~\cite{UganderBK13} specifically use 4-vertex pattern counts to 
provide a ``map" of egonets, and show significant patterns among these counts. 
Such analyses require fairly precise frequency counts. 

\subsection{Related Work} \label{sec:related}

Motif counting for bioinformatics was arguably initiated by a 
seminal paper of Milo {\em et al.}~\cite{Milo2002}. This technique has been
used for graph modeling~\cite{PrzuljCJ04,HoBe+07}, graph comparisons~\cite{PrzuljCJ04,HalesA08}, and even decomposing a network~\cite{Itzkovitz_2005}.
Refer to \cite{BetzlerBFKN11,WongBQH12} for more details. 

Triangle counting has a rich history in social sciences and related analyses, that
we simply refer the reader to the related work sections of~\cite{TsKoMi11,SePiKo13}.
The significance of 4-vertex patterns was studied in recent work of Ugander et al.~\cite{UganderBK13},
who propose a ``coordinate system'' for graphs based on the motifs distribution. This is used
for improved network classification, and the input graphs were comparatively small (thousands of vertices).

Previous studies tailored to 4-vertex patterns~\cite{GonenS09,MarcusS10} provide both exact and approximation algorithms. However, the asymptotic bounds in these graphs are far from practical, and they are only applied to small graphs. For example, a graph with 90K edges requires 40 minutes of processing~\cite{MarcusS10}. Color coding~\cite{AlYuZw94}, Monte-Carlo Markov Chain sampling~\cite{BhRaRa+12}, and edge sampling to speedup edge iteration based algorithms~\cite{RaBhHa14} have been adopted to count patterns in graphs. We will provide `detailed comparisons with these methods in~\Sec{compare}.   To exploit more powerful computing platforms, incremental pattern building algorithms
for Map-Reduce have been described in~\cite{Pl12,BhHa13}. 

Most relevant to this work are previous studies on \emph{wedge sampling}~\cite{ScWa05-2,SePiKo13,KoPiPlSe13}. This method samples paths of length 2 to estimate
various triangle statistics. Our method of 3-path sampling can be seen as building
on wedge sampling. We employ new path pruning techniques to improve the algorithm's efficiency. 
These pruning techniques are inspired by degeneracy ordering algorithms
for triangle counting~\cite{ChNi85,ScWa05}. We can actually provide mathematical error bars for real
runs and instances (as opposed to just a theoretical proof of convergence of estimate).

\subsection{Summary of our contributions}
We design a new randomized algorithm, based on \emph{3-path sampling}, that outputs accurate
estimates of all 4-vertex subgraphs counts. The algorithm is provably correct and makes
no distributional assumption on the graph. All probabilities are over the internal randomness
of the algorithm itself (which is independent of the instance). 
We run detailed simulations on a large variety of datasets, including product co-purchasing networks,
web networks, autonomous systems networks, and social networks. All experiments are
done on a single commodity machine using 64GB memory.

\begin{asparaitem}%[\IEEEsetlabelwidth{This must be good}]
\item[\textbf{Extremely fast.}] Our algorithm relies on a sampling based approach making it extremely fast even on very large graphs. Indeed, there are instances where a finely tuned enumeration code takes almost a day to compute counts of 4-vertex motifs whereas our algorithm only takes less than a minute to output accurate estimates. Refer to \Fig{speedup} for speedup over a well-tuned enumeration code.
Our algorithm takes less a minute on an Orkut social network with 200 million edges, where the total count of each motif is over a billion (and most counts are over 10 billion). 
An input Flickr social network has more than 10 billion 4-cliques; we get estimate of this 
number with less than $0.5\%$ error within 30 seconds on a commodity machine. We do not preprocess any of the graphs,
and simply read them as a list of edges.

\item[\textbf{Excellent empirical accuracy.}] We empirically validate our algorithm on a large
variety of datasets, and it consistently gives extremely accurate answers. Refer to \Fig{error}.
We get $< 1\%$ relative error \emph{for all subgraph counts on all datasets}, even those with more than 100M edges. (Exact counts were obtained by brute-force enumerations that took several days.)
This is much more accurate than any existing method to count such motifs.
We compare with existing sampling methods, and demonstrate that our algorithm is faster and more accurate than
the state-of-the-art.

\item[\textbf{Provable guarantees with error bars.}] Our algorithm has a provable guarantee
on accuracy and running time. Furthermore, we can quantify the accuracy/confidence on real
inputs and runs of our algorithm. For a given number of samples, we can have a method to 
put an explicit error bar on our estimate, based on asymptotically tight versions of Chernoff's bound.
While these error bars are not as tight as the real errors in \Fig{error}, we can still
mathematically prove that the errors are mostly within 5\% and always within 10\%.

\item[\textbf{Trends in 4-vertex pattern counts:}] Given the rapid reporting of 4-vertex
pattern counts, our algorithm can be used as a tool for motif analysis. We detect common
trends among a large variety of graphs. Not surprisingly, the 3-star is the most frequent
4-vertex motif in all graphs we experimented upon. The least frequent is either the 4-cycle
or the 4-clique. The chordal-4-cycle frequency is always more than that of the 4-cycle or 4-clique.
Ugander et al~\cite{UganderBK13} study what trends are merely implied by graph
theory, and what are actually features of real-world graphs. Such analyses require
accurate estimates quickly, which our algorithm can provide. It is a promising direction to use
our algorithm to provide more input to such studies.
\end{asparaitem}

\section{Formal description of the problem}\label{sec:notation}

Our input is an undirected simple graph $G=(V,E)$, with $n$ vertices and $m$ edges.
For vertex $v$, $d_v$ is the degree of $v$.

It is critical to distinguish subgraphs from \emph{induced subgraphs}.
A subgraph is simply some subset of edges. An induced subgraph is obtained
by taking a subset $V'$ of vertices, and consider \emph{all edges}
among these vertices. Refer to \Fig{example-graph}. The edges $(v_1,v_2),(v_2,v_3),(v_3,v_4),(v_4,v_1)$
form a 4-cycle, but the vertex set $\{v_1,v_2,v_3,v_4\}$ induces a chordal-4-cycle.
We collectively refer to the 4-vertex subgraphs as ``motifs".

It is technically convenient to think of induced subgraph counts. We 
denote the number of induced occurrences of the $i$-th subgraph (of \Fig{4-vertex-motifs})
by $\ind_i$. So, $\ind_4$ is the number of induced 4-cycles in $G$,
which is the number of distinct subsets of 4 vertices that induce a 4-cycle.
When we talk of a ``vanilla" subgraph, we mean the usual subgraph setting
(a subset of edges). In general, if we do not say ``induced", we mean
vanilla.

\emph{Our aim is to get an estimate of all $\ind_i$ values.} Let $\nonind_i$
denote the number of (vanilla) subgraph occurrences of the $i$th subgraph,
There is a simple linear relationship between induced and non-induced counts,
given below. The $(i,j)$ entry of the matrix $A$ below 
is simply the number of distinct copies of the $i$th
subgraph in the $j$th subgraph (so $A_{2,4} = 4$, the number of 3-paths
in the 4-cycle).
\begin{align}
\begin{split}
\begin{pmatrix}
1 & 0 & 1 & 0 & 2 & 4 \\
0 & 1 & 2 & 4 & 6 & 12 \\
0 & 0 & 1 & 0 & 4 & 12 \\
0 & 0 & 0 & 1 & 1 & 3 \\
0 & 0 & 0 & 0 & 1 & 6 \\
0 & 0 & 0 & 0 & 0 & 1 
 \end{pmatrix} 
 \cdot \begin{pmatrix} \ice_1 \\ \ice_2 \\ \ice_3 \\ \ice_4 \\ \ice_5 \\ \ice_6 \end{pmatrix} 
 &= 
  \begin{pmatrix} \nice_1 \\ \nice_2 \\ \nice_3 \\ \nice_4 \\ \nice_5 \\ \nice_6 \end{pmatrix}. 
  \end{split}
 \label{eq:matrix}
 \end{align}

\begin{figure}
    \centering
    	\includegraphics[scale=0.75]{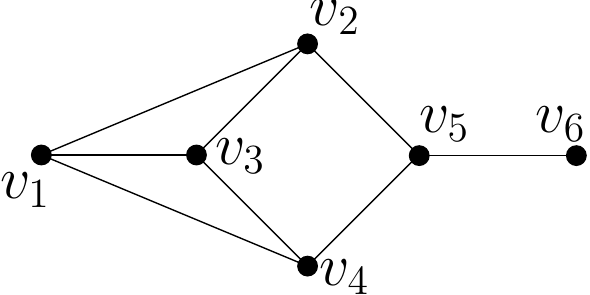}
  \caption{An example graph.}
  \label{fig:example-graph}
\end{figure}

\section{The basic algorithm: estimating counts via 3-path sampling}

Our algorithm for estimating counts of 4-vertex motifs is based on 3-path sampling. 
In this section, we discuss a basic version of this method. In the next section,
we enhance it to improve accuracy.

We begin with a simple procedure that samples a uniform (vanilla) random 3-path.
For each edge $e=(u,v) \in E$, denote $\tau_e = (d_u-1)(d_v-1)$.
We denote $W = \sum_e \tau_e$.

\begin{algorithm}
 \caption{\sample}\label{alg:sample}
 \DontPrintSemicolon
 Compute $\tau_e$ for all edges and set $p_e = \tau_e/W$.\;
 Pick edge $e = (u,v)$ with probability $p_e$.\label{step2}\;
 Pick uniform random neighbor $u'$ of $u$ \emph{other than $v$}.\;
 Pick uniform random neighbor $v'$ of $v$ \emph{other than $u$}.\;
 Output the three edges $\{(u',u),(u,v),(v,v')\}$.
\end{algorithm}

Observe that the output of \sample{} can either be a triangle (if $u' = v'$) or a 3-path. 

\begin{claim} \label{clm:sample} Fix any 3-path. The probability
that \sample{} outputs this 3-path is exactly $1/W$.
\end{claim}

\begin{proof} Fix a 3-path $(u',u),(u,v),(v,v')$ ($u, u', v, v'$ are all distinct). 
The probability that $e = (u,v)$ is selected as the middle edge (in Step~\ref{step2})
is exactly $(d_u - 1)(d_v-1)/W$. Conditioned on this event, the probability
that $u'$ is selected as a neighbor of $u$ is $1/(d_u-1)$ (note that the neighbor $v$
is excluded). Similarly, $v'$ is selected with probability $1/(d_v-1)$.
Putting it all together, the 3-path is chosen with 
probability $[(d_u - 1)(d_v-1)/W] \cdot [1/(d_u-1)] \cdot [1/(d_v-1)] = 1/W$.
The probability is the same for all 3-paths, proving our claim.
\end{proof}

All motifs of \Fig{4-vertex-motifs}, except the 3-star,
contain a 3-path. So one can perform the following experiment. Run \sample{}
to get a collection of edges, and hence a set of (at most $4$) vertices.
Check the motif induced by this set of vertices.
Repeat this experiment many times to 
estimate the true counts $\ind_i$ ($i \in [2,6]$).
Finally, we use the formula of \Eqn{matrix} to estimate $\ind_1$.
This is exactly the algorithm \estimate, as given in \Alg{basic}.
We remind the reader that $A_{2,i}$ is the number of 3-paths in
the $i$th motif.

\begin{algorithm}
 \caption{\estimate \newline Input: graph $G = (V, E)$, samples $\samp$}\label{alg:basic}
 \DontPrintSemicolon
 Run \sample{} $k$ times to get $k$ sets of edges. Let $S_\ell$ denote
 the set of corresponding vertices for the $\ell$th set.\;
 Initialize $count_i = 0$ for $i \in [2,6]$.\;
 For $\ell \in [1,k]$,\;
 \ \ \ \ Determine subgraph induced by $S_\ell$.\;
 \ \ \ \ If this is the
 $i$th motif, increment $count_i$.\;
 For each $i \in [2,6]$,\;
 \ \ \ \ Set $\widehat{\ind}_i = (count_i/\samp)\cdot(W/A_{2,i})$.\;
 Set $\nonind_1 = \sum_{v \in V} {d_v \choose 3}$.\;
 Set (induced 3-stars) $\widehat{\ind}_1 = \nonind_1 - \widehat{\ind}_3
 - 2\widehat{\ind}_5 - 4\widehat{\ind}_6$.
 \end{algorithm}

We prove that \estimate{} outputs unbiased estimates for all $\ind_i$s. 

\begin{theorem} \label{thm:exp}
For every $i \in [1,6]$, $\EX[\aice_i] = \ice_i$.
\end{theorem}
\begin{proof} First, let use deal with subgraphs other than the 3-star, so fix some $i \neq 1$.
For each $\ell \in [k]$, let $X_\ell$ be the indicator random variable for $S_\ell$ inducing
the $i$th motif. So $X_\ell$ is $1$ iff the $\ell$th call to \sample{} outputs
a 3-path contained in a copy of the $i$th motif. The total number of (distinct) such 3-paths
is exactly $A_{2,i}\cdot\ind_i$. By \Clm{sample}, the probability that $X_\ell = 1$
is $A_{2,i}\cdot\ind_i/W$. Hence, $\EX[X_\ell] = \ind_i\cdot A_{2,i}/W$.

We have $\EX[count_i] =$ $\sum_{\ell = 1}^\samp \EX[X_\ell]$
$= (\samp\ind_i A_{2,i})/W$, by linearity of expectation. Hence, $\EX[\aice_i] = \ind_i$.
Now, we detail with $\aice_1$. Note that $N_1$, the number of 3-star subgraphs,
is exactly $\sum_{v \in V} {d_v \choose 3}$.
By linearity of expectation,
$\EX[\widehat{\ind}_1] = \nonind_1 - \EX[\widehat{\ind}_3]
 - 2\EX[\widehat{\ind}_5] - 4\EX[\widehat{\ind}_6]$,
 which is $\nonind_1 -{\ind}_3 - 2{\ind}_5 - 4{\ind}_6$ $=\ind_1$
(as given by \Eqn{matrix}).
\end{proof}

We can also prove concentration results using the Hoeffding bound~\cite{Ho63}. This is useful as a proof
of concept, but does not give useful bounds in practice. (We give more details
later.) This analysis
is analogous to that of wedge sampling results~\cite{ScWa05,SePiKo13}. 

\begin{theorem}[Hoeffding \cite{Ho63}]
  \label{thm:hoeff}
  Let $X_1, X_2, \dots, X_k$ be independent random variables with $0
  \leq X_i \leq 1$ for all $i=1,\dots,k$.  Define $\bar X =
  \frac{1}{k} \sum_{i=1}^k X_i$. Let $\mu = \EX[\bar X]$. 
  Then for $\eps \in (0,1)$, we have
  \begin{displaymath}
    \Pr[ |\bar X - \mu | \geq \eps ] \leq 2 \exp(-2 k \eps^2).
  \end{displaymath}
\end{theorem}

We can derive concentration results quite directly from this bound.
Note that the bound is for a fixed $i \in [2,6]$, i.e., a fixed motif that is
not the $3$-star.

\begin{theorem} \label{thm:basic-conc} Fix $\delta, \eps \in (0,1)$ and $i \in [2,6]$.
Set $\samp = \ceil{(2\eps)^{-2}\ln(2/\delta)}$. For all $i \in [2,6]$:
with probability at least $1-\delta$, $|\aice_i - \ind_i| < \eps W/A_{2,i}$.
With probability at least $1-\delta$, $|\aice_1 - \ind_1| < \eps W$.
\end{theorem}

\begin{proof} We have fixed $i \neq 1$. 
For each $\ell \in [k]$, define $X_\ell$ to be the indicator random variable for $S_\ell$ inducing
the $i$th motif. Observe that \emph{when $i$ is fixed}, each $X_\ell$ is independent,
since it is simply the result on an independent sample. In other words, the chance of the $\ell$th sample
inducing the $i$th motif is independent over the sample index $\ell$. 

Apply \Thm{hoeff} to $X_1, \ldots, X_\samp$. With probability at most $\delta$,
$|\bar{X} - \EX[\bar{X}]| \geq \eps$ (we use the notation from \Thm{hoeff}).
It remains to interpret $\bar{X}$. Note that $\aice_i = (count_i/\samp)\cdot(W/A_{2,i})$.
Since $count_i = \sum_{\ell=1}^k X_\ell$, $\aice_i = \bar{X}\cdot(W/A_{2,i})$.
So $|\bar{X} - \EX[\bar{X}]| \geq \eps$ implies $|\aice_i - \ind_i| \geq \eps$,
as desired.

Since $\aice_1$ is obtained by subtracting out other terms, it appears that the
errors could add up. With a little care, we can get the same bound as the other $\aice_i$'s.
Define random variable $Y_\ell$ as follows: if $S_\ell$ induces a tailed triangle,
$Y_\ell = 1/A_{2,3} = 1/2$. If $S_\ell$ induces a chordal-4-cyle, $Y_\ell = 2/A_{2,5} = 2/6$.
If $S_\ell$ induces a 4-clique, $Y_\ell = 4/A_{2,6} = 4/12$. In all other cases, $Y_\ell = 0$.
We have constructed this random variable, so that $\EX[Y_\ell] = (\ind_3 + 2\ind_5 + 4\ind_6)/W$.

Observe that $\aice_1$ can also be expressed as $N_1 - (\sum_\ell Y_\ell/\samp)W$. The additive
error $|\aice_1 - \ind_1|$ is the same as $W\cdot|\bar{Y} - \EX[\bar{Y}]|$. Applying
\Thm{hoeff}, with probability at least $1-\delta$, $|\aice_1 - \ind_1| < \eps W$.
\end{proof}

To get error bounds for obtaining all estimates, we take the union bound. 
Hence, with probability at least $1-6\delta$, we get the same accuracy guarantee of \Thm{basic-conc} simultaneously for all counts.

\begin{theorem} \label{thm:perf-samp} The running time and total storage (including all preprocessing) of \estimate{} is $O(m+\samp)$.
\end{theorem}

\begin{proof} We will assume that the adjacency lists are stored in standard hash tables, to allow for constant lookup time.
(Hence, our running time is actually expected over the hash table. If we store the adjacency list in a tree-based
dictionary data structure, we incur an additional factor of $\log (\max_v d_v)$ in the running time.)
We can access the degree $d_v$ and a random neighbor of $v$ in constant time.

The preprocessing required to determine each value of $\tau_e$ is linear in $m$. Generating a single
$3$-path sample takes $O(1)$ time. Checking the motif induced by the path also takes $O(1)$ time.
The only additional storage are the values $\tau_e$ and the various counts. Hence, both the running
time and storage can be bounded by $O(m+\samp)$.
\end{proof}

\subsection{The challenge of cycle-based motifs} \label{sec:cycle}

\Thm{exp} and \Thm{basic-conc} seem to give us all we want, so why aren't we done? The catch is that the concentration
bound of \Thm{basic-conc} is actually too weak to give reasonable estimates for real world graphs. 
Let us do some rough calculations, ignoring the constants. To get an estimate such that $|\aice_i - \ind_i| < \eps W$,
we require $\samp \approx 1/\eps^2$. But for such an estimate to be useful, we need to understand how $W$ relates
to $\ind_i$. So $\eps$ needs to be of the order of $\ind_i/W$, and consequently, $\samp$ needs to be $(W/\ind_i)^2$.

\begin{table}[h]
\centering
\caption{$W$ vs $\ind_i$: counts given as orders of magnitude.}\label{tab:comp}
\begin{tabular}{l||c|c|c|c|c}
\hline
Graph &  $W$ &  $C_2$  &  $C_3$ & $C_4$ & $C_6$ \\ \hline
amazon0312	& E+09	& E+08 & E+08	& E+06 & E+06 \\
as-skitter	 & E+12 & 	E+11  & E+11	& E+10 & E+08\\
orkut	 & E+13	& E+13& E+12 & E+10 & E+09 \\ \hline
\end{tabular}
\end{table}

Refer to \Tab{comp} for the values of $W$ and a few $C_i$s. (For convenience, we just give
the order of magnitude of each number. Full numbers are given later.)
For $i \in \{1,2,3\}$ (3-star, 3-path, and tailed triangle),
$(W/\ind_i)^2$ is usually $< 10^4$. This is fairly reasonable number of samples to take, and 
leads to an efficient and accurate algorithm. On the other hand, for $i \in \{4,5,6\}$ (4-cycle, chordal-4-cycle, and 4-clique),
$(W/\ind_i)^2$ is often $> 10^8$, which is too many samples to take.

In other words, \estimate{} does not perform well for motifs containing a 4-cycle. This leads us to a new algorithm
for dealing with these motifs, as described in the next section.

\section{Improved estimation of 4-cycle-based motifs  via {\em centered} 3-paths}

We denote the 4-cycle, chordal-4-cycle, and 4-clique as {\em cycle-based} motifs. We design a better algorithm to estimate
them. While the algorithm is provably correct for any graph, the fact that it gives a significant improvement is dependent
on the structure of real-world graphs. 

Our aim is to find a subset $\cS$ of 3-paths with the following properties:
\begin{compactitem}
	\item Every cycle-based motif is guaranteed to contain a fixed number of 3-paths from $\cS$.
	\item It is possible to quickly generate uniform random samples from $\cS$.
	\item $|\cS|$ is significantly smaller than $W = \sum_{e = (u,v)} (d_u-1)(d_v-1)$.
\end{compactitem}

Let us go back to \sample, and think of enumerating all 3-paths. For edge $(u,v)$,
we take every neighbor of $u$ and every neighbor of $v$ to generate a 3-path.
We basically take the Cartesian product of the adjacency lists of $u$ and $v$.
Could we prune the adjacency lists so this product is smaller? 

Suppose we order all vertices based on degree and vertex id.
So we say $u \prec v$ if: $d_u < d_v$ or, if $d_u = d_v$, the vertex id of $u$
is less than that of $v$.
We could prune the lists using this ordering. When looking for 3-paths where
$(u,v)$ is the middle edge, we only look at the portion of $u$'s list ``greater"
than $v$, and the portion of $v$'s list greater than $u$. In general, many 3-paths
are generated when $d_u$ and $d_v$ are large. But in that case, we hope that many neighbors
of $u$ and $v$ are of lower degree. The pruning ignores such vertices and (hopefully)
reduces the set of 3-paths considered. Let us define the set $\cS$ of centered 3-paths.
\begin{definition}[Centered 3-path]\label{def:canonical-3-path} 
A 3-path formed by edges $\{(t,u),(u,v),(v,w)\}$ is \emph{centered} 
if: $v \prec t$, $u \prec w$, and the edge $(t,w)$ exists in the graph
(so $t,u,v,w$ form a 4-cycle).
\end{definition}
We prove the important property that every cycle-based motif contains
a fixed number of centered 3-paths.

\begin{lemma}
\label{lem:unique-3path}
Every induced 4-cycle and chordal-4-cycle contains exactly one centered 3-path. Every induced 4-clique contains
exactly three centered 3-paths.
\end{lemma}

\begin{proof} Consider any (vanilla) 4-cycle, formed by vertices (in order) $t,u,v,w$. 
Pick the smallest vertex, say $u$. Pick the neighbor of $u$ that is smaller, say $v$.
We show that the 3-path $\{(t,u),(u,v),(v,w)\}$ is the only centered 3-path in this 4-cycle.

By the choice of $(u,v)$, $v \prec t$ and $u \prec w$. Hence, we see that $\{(t,u),(u,v),(v,w)\}$
is centered. The only other possible centered 3-path is $\{(u,t),(t,w),(w,v)\}$.
Because $v \prec t$, this path cannot be centered. That completes
the proof for the induced 4-cycle case.

Now, suppose $t,u,v,w$ forms an induced chordal-4-cycle. The extra 3-paths contain
the chord in the middle, and such 3-paths do not lie on a 4-cycle. So there
only exists one centered 3-path.

A 4-clique contains three 4-cycles that partition the 12 different 3-paths. Each of these 
4-cycles has a centered 3-path, yielding a total of three such 3-paths.
\end{proof}

We now show how to sample a uniform random centered 3-path. It is quite analogous 
to \sample. First, some notation. Let $L_{u,v}$ be the number of neighbors of $u$ 
greater than $v$.
By sorting all adjacency lists according to vertex degree and id, 
we can compute for every edge $e = (u,v)$, the value $\lambda_e = L_{u,v}L_{v,u}$.
Let $\Lambda = \sum_e \lambda_e$.

\begin{algorithm}
 \caption{\samplecent}\label{alg:samplecent}
 \DontPrintSemicolon
 Compute $\lambda_e$ for all edges and set $p_e = \lambda_e/\Lambda$.\;
 Pick edge $e = (u,v)$ with probability $p_e$.\label{step2cent}\;
 Pick uniform random neighbor $u'$ of $u$ such that $v \prec u'$.\;
 Pick uniform random neighbor $v'$ of $v$ such that $u \prec v'$.\;
 Output the three edges $\{(u',u),(u,v),(v,v')\}$.
\end{algorithm}

Note that it is possible that \samplecent{} outputs a 3-path
that is not centered (if the 3-path does not lie on a 4-cycle).
Nonetheless, analogous to \Clm{sample}, we have the following.

\begin{claim} \label{clm:samplecent} Fix any centered 3-path. The probability
that \samplecent{} outputs this 3-path is exactly $1/\Lambda$.
\end{claim}

Now, we give the algorithm that estimates the number of cycle-based motifs. It is
analogous to \estimate, only using centered 3-paths. For convenience,
let $B_i$ denote the number of centered 3-paths in the $i$th motif.
So $B_4 = B_5 = 1$ and $B_6 = 3$, by \Lem{unique-3path}.

\begin{algorithm}
 \caption{\estcent \newline Input: graph $G = (V, E)$, samples $\samp$}\label{alg:centered}
 \DontPrintSemicolon
 Run \samplecent{} $k$ times to get $k$ set of edges. Let $T_\ell$ denote
 the set of corresponding edges for the $\ell$th set.\;
 Initialize $count_i = 0$ for $i \in [4,6]$.\;
 For $\ell \in [1,k]$,\;
 \ \ \ \ If $T_\ell$ is a centered 3-path,\;
 \ \ \ \ \ \ \ Determine subgraph induced by $S_\ell$.\;
 \ \ \ \ \ \ \ If this is the
 $i$th motif, increment $count_i$.\;
 For each $i \in [4,6]$,\;
 \ \ \ \ Set $\widehat{\ind}_i = (count_i/\samp)\cdot(\Lambda/B_i)$.\;
\end{algorithm}

Analogous to \Thm{basic-conc}, we can prove the following.
Observe how $W$ is replaced by $\Lambda$.

\begin{theorem} \label{thm:cent-conc} Fix $\delta, \eps \in (0,1)$
and set $\samp = \ceil{(2\eps)^{-2}\ln(2/\delta)}$. For all $i \in [4,6]$:
with probability at least $1-\delta$, $|\aice_i - \ind_i| < \eps \Lambda/B_i$.
\end{theorem}

For the same number of samples, the performance of \estcent{} requires an additional logarithmic factor because of additional preprocessing.
In general, $d_v$ is much smaller than $n$ (and is effectively constant for most vertices), so the additional logarithmic factor
is not too expense. We require fewer samples for the same accuracy, so \estcent{} wins at scale.

\begin{theorem} \label{thm:perf-cent} The running time of \estcent{} is $O(\sum_v d_v\log d_v+\samp)$ and the total storage
is $O(m+\samp)$.
\end{theorem}

\begin{proof} As discussed earlier, we need to sort the adjacency lists to determine each value of $\lambda_e$.
Once the lists are sorted, each centered $3$-path sample can generated in constant time.
That leads the running time bound. The remaining analysis is identical to that of \Thm{perf-samp} for \estimate.
\end{proof}

\subsection{Why centered 3-paths help} \label{sec:help}

We put the value of $W$ and $\Lambda$ for various real world networks in \Tab{centered}.
Observe how $\Lambda$ is at least an order of magnitude smaller than $W$. 
This is a huge difference when it comes to the sampling bounds
in \Thm{basic-conc} and \Thm{cent-conc}. These bounds show that
\emph{two orders of magnitude less samples} suffice for the same error (in estimating
cycle-based motifs). This improvement is extremely significant for getting
good accuracy with fewer samples.

\begin{table}[h]
\centering
\caption{Difference between the number of 3-paths and the number of centered 3-paths.}\label{tab:centered}
\begin{tabular}{l||c|c|c}
\hline
Graph &  $W$ &  $\Lambda$  &  $W/\Lambda$ \\ \hline
amazon0312	& 1.40E+09	& 9.36E+07	& 15 \\
amazon0505	& 1.59E+09 & 	1.02E+08	& 16 \\
amazon0601	& 1.57E+09 & 	1.01E+08	 & 15 \\
as-skitter	 & 1.43E+12 & 	9.05E+10	& 16\\
cit-Patents & 	9.16E+09	& 8.78E+08 & 	10 \\
web-BerkStan & 	1.69E+12	& 1.28E+11 & 	13\\
web-Google & 	2.05E+10	& 6.34E+08	& 32 \\
web-Stanford & 	1.85E+11	& 1.36E+10 & 	14 \\
wiki-Talk	& 1.31E+12 &  9.08E+09	& 144 \\
youtube & 	1.19E+11	& 1.68E+09 & 	71\\
flickr	& 1.31E+13 & 	8.42E+11	& 16 \\
livejournal	 & 1.67E+12	& 1.14E+11	& 15 \\
orkut	 & 2.22E+13	& 9.48E+11	& 23 \\ \hline
\end{tabular}
\end{table}

The final algorithm is simply obtained by running both \estimate{}
and \estcent. The former gives estimates for $C_1, C_2, C_3$ (we simply discard
the remaining estimates), and the latter estimates $C_4, C_5, C_6$.

\section{Getting practical error bars} \label{sec:error}

While the Hoeffding bound used above provides theoretical convergence, we do not
get practical error bars from it. In this section, we show how to get useful
error bars for our algorithm on real instances. 

All of our sampling algorithms have the same underlying primitive: try to estimate
the expectation $p$ of a Bernoulli random variable. We 
generate a binomial random variable $X \sim B(\samp,p)$ (by performing $\samp$ i.i.d.
Bernoulli trials), and hope that the outcome is close enough to the expectation. 

We employ a standard Bayesian viewpoint to generate
an error bar. Suppose, our outcome of the binomial draw is $X = r$. 
Conditioned on a choice of $p$, we calculate the probability that $X=r$.
This gives a prior
on $p$. Of course, this cannot
be done explicitly because of computational issues, but we can use tail bounds for $B(\samp,p)$
to get appropriate estimates. We use the following theorem of Chernoff~\cite{Che52}
(we use notation of Equation 1.4 from~\cite{DuPa09}) that gives good tail bounds for $B(\samp,p)$.

\begin{theorem} [Chernoff] Suppose $X \sim B(\samp,p)$. Fix $\alpha \in (0,1)$.
\begin{align*}
\Pr[X/\samp \geq \alpha] \leq \exp(-D(\alpha,p) \samp) \ \ \textrm{if $\alpha > p$}\\
\Pr[X/\samp \leq \alpha] \leq \exp(-D(\alpha,p) \samp) \ \ \textrm{if $\alpha < p$}
\end{align*}
where $D(a,b) = a\ln(a/b) + (1-a)\ln((1-a)/(1-b))$ (the KL-divergence
between Bernoulli distributions with expectation $a$ and $b$).
\end{theorem}

Suppose the outcome of $X/\samp = \alpha$. We can use the Chernoff bound to get
a range of likely values of $p$. Think of $\exp(-D(\alpha,p) \samp)$ as a function of $p$.
The basic properties of the KL-divergence (and simple algebra) imply that $\exp(-D(\alpha,p) \samp)$
is a unimodal function with a maximum value of $1$ at $p = \alpha$
and a minimum of $0$ at $p = 0,1$. That motivates the following
definition.

\begin{definition} \label{def:p} Fix $\samp,\alpha, x \in (0,1)$. Then $p_l(\samp,\alpha,x)$ (lower)
and $p_u(\samp,\alpha,x)$ (upper) are the two unique values of $p$ such that 
$\exp(-D(\alpha,p) \samp) = x$.
\end{definition}

With this definition, we can give precise error bars. In other words, given
the outcome of a binomial random variable $B(\samp,p)$, we can give an interval
of plausible values (up to any desired confidence $\delta$) for $p$.

\begin{corollary} [Error bar for binomial distribution] \label{cor:cher} Fix binomial
distribution $B(\samp,p)$, and $\alpha, \delta \in (0,1)$
Then, for any $p \notin [p_l(\samp,\alpha,\delta), p_u(\samp,\alpha,\delta)]$,
$$ \Pr_{X \sim B(\samp,p)}[X/\samp = \alpha] \leq \delta$$
\end{corollary}

How does this relate to our algorithms? Observe that in both \Alg{basic} and \Alg{centered},
the variables $count_i$ are binomial random variables. So we can produce errors
bars for $count_i/\samp$ using the above corollary.
The final estimates are of
the form $\aice_i = (count_i/\samp)\cdot K_i$ ($i \neq 1$, and $K_i$ is some fixed scaling, depending
on the algorithm and $i$). So error bars for $count_i/\samp$ directly translate to error
bars for $\aice_i$. For $i=1$ (3-stars), we simply add up the errors
(in \estimate{}) for $\aice_3$, $2\aice_5$, and $4\aice_6$.

\section{Experimental Results}

\textbf{Preliminaries:} We implemented our algorithms in {\tt C} and ran our experiments on a
computer equipped with a 2x2.4GHz Intel Xeon processor with 6~cores
and  256KB  L2 cache (per core), 12MB L3 cache, and  64GB memory. 
We performed our experiments on 13 graphs  from SNAP~\cite{Snap}.
In all cases, directionality is ignored, and duplicate edges are omitted. 
\Tab{enumeration} has the properties of these graphs, where
$|V|$ and $|E|$  are the numbers of vertices and edges, respectively.  

Exact counts for the motifs are obtained by a well-tuned enumeration (counts and runtime given in \Tab{enumeration}).
This algorithm only enumerates the 4-cycle, the chordal-4-cycle, and the 4-clique, and uses direct approaches
to get other counts. For convenience, we refer to this the \emph{enumeration code}.
We do not get into details, but note that this code processes million edge 
Amazon networks in only 5 seconds\footnote{This is quite competitive with the best existing numbers in the literature of~\cite{MarcusS10}, whose algorithm takes 40 minutes on a 90K autonomous systems graph.}. It uses vertex orderings for speedup, analogous to using degeneracy orderings for triangle
enumeration~\cite{ChNi85,ScWa05}.
For getting 3-path sampling estimates, we run both \estimate{} and \estcent{} as described earlier, with $\samp = 200K$. 
We use the outputs of $\aice_1, \aice_2, \aice_3$ as given by \estimate{}, and $\aice_4, \aice_5, \aice_6$ from \estcent.
The runtimes are in the last column of \Tab{enumeration}.

\begin{table*}[t]
\centering
\caption{Exact values of pattern counts and runtimes (in seconds).}\label{tab:enumeration}
\hspace*{-2pt}
{\small
\begin{tabular}{|l|c|c|c|c|c|c|c|c|r|r|}
\hline
Datasets & $|V|$ & $|E|$& 3-star & 3-path & Tailed & 4-cycle & Chordal & 4-clique & {\scriptsize Enum.} & {\scriptsize 3-path} \\ 
& & & & & triangle &&4-cycle&& {\scriptsize time} & {\scriptsize  time} \\\hline
amazon0312 & 4.01{\scriptsize  E+}5 & 2.35{\scriptsize  E+}6 & 1.07E+10 & 8.44E+08 & 1.90E+08 & 3.23E+06 & 1.71E+07 & 3.98E+06 & 4.42 &0.47 \\ \hline
amazon0505 & 4.10{\scriptsize E+}5	& 2.44{\scriptsize  E+}6 & 1.21E+10 & 9.63E+08 & 2.19E+08 & 3.30E+06 & 1.91E+07 & 4.36E+06 & 4.75 & 0.48 \\ \hline
amazon0601 & 4.03{\scriptsize E+}5	& 2.44{\scriptsize  E+}6 &1.11E+10 & 9.41E+08 & 2.17E+08 & 3.22E+06 & 1.92E+07 & 4.42E+06 & 4.74 & 0.48 \\ \hline
as-skitter & 1.70{\scriptsize E+}6	& 1.11{\scriptsize  E+}7 & 9.64E+13 & 8.19E+11 & 1.62E+11 & 4.27E+10 & 1.96E+10 & 1.49E+08 & 5128.93 & 2.7 \\ \hline
cit-Patents & 3.77{\scriptsize E+}6	& 1.65{\scriptsize  E+}7 & 6.11E+9 & 6.54E+09 & 5.52E+08 & 2.69E+08 & 6.28E+07 & 3.50E+06 & 46.46 & 3.33\\ \hline
flickr & 1.86{\scriptsize E+}6	& 1.56{\scriptsize  E+}7 & 1.90E+13 & 6.89E+12 & 1.18E+11 & 1.18E+11 & 2.30E+11 & 2.67E+10 & 217274.39 & 2.53\\ \hline
livejournal & 5.28{\scriptsize E+}6 &	4.87{\scriptsize  E+}7&  4.46E+12 & 1.14E+12 & 1.26E+11 & 5.21E+09 & 1.90E+10 & 1.14E+10 & 11894.63 & 6.86\\ \hline
orkut & 3.07{\scriptsize E+}6 & 	2.24{\scriptsize  E+}8 & 9.78E+13 & 1.86E+13 & 1.51E+12 & 7.01E+10 & 4.78E+10 & 3.22E+09 & 70966.96 & 16.24\\ \hline
web-BerkStan &6.85{\scriptsize E+}5	& 6.65{\scriptsize  E+}6 &  3.82E+14 & 3.14E+10 & 4.76E+11 & 2.53E+10 & 9.86E+10 & 1.07E+09 & 6462.56 &3.77 \\ \hline
web-Google & 8.76{\scriptsize E+}5	& 4.32{\scriptsize  E+}6&  6.50E+11 & 4.06E+09 & 6.72E+09 & 3.80E+07 & 3.82E+08 & 3.99E+07 & 52.29 & 0.88\\ \hline
web-Stanford &  2.82{\scriptsize E+}5	& 1.99{\scriptsize  E+}6& 2.51E+13 & 1.28E+10 & 5.08E+10 & 4.48E+09 & 8.60E+09 & 7.88E+07 & 831.50 & 1.76\\ \hline
wiki-Talk & 2.39{\scriptsize E+}6	& 4.66{\scriptsize  E+}6 & 1.92E+14 & 1.17E+12 & 6.41E+10 & 9.24E+08 & 1.03E+09 & 6.49E+07 & 1346.76 & 1.04\\ \hline
youtube &  1.16{\scriptsize E+}6	& 4.95{\scriptsize  E+}6&  5.73E+12 & 9.15E+10 & 1.24E+10 & 2.32E+08 & 2.22E+08 & 4.99E+06 & 141.78 & 0.61 \\ \hline
\end{tabular}
}
\end{table*}

\textbf{Convergence of estimates:} To show  convergence, we perform detailed runs on the as-skitter graph.
We choose this because it is the most difficult to get accurate estimates, since the cycle-based motif
counts are small relative to the graph size. We vary the numbers of samples in increments of 2.5K. For each
choice of the number of samples, we perform 50 runs of our algorithm. 
We plot those results in \Fig{convergence}
for tailed-triangles, chordal-4-cycles, and 4-cliques. (Other patterns are omitted due the space considerations,
and had even better convergence.) The output of each run (for a given number of samples) is depicted by a blue dot.
For 4-clique counts, we can see the spread of outputs reducing. The figure only goes
up to 35K samples. (The convergence is so rapid that at around 50K samples, the spread is impossible to see.)

\textbf{Accuracy:}  \Fig{error} presents the relative errors for all 13 graphs and all 6 motifs, using 200K samples for 
both \estimate{} and \estcent.  \emph{All relative errors are less than 1\% in all instances.} As expected the relative errors tend to be larger for the less frequent patterns such as 4-cycles and 4-cliques.

\begin{figure*}[t]
    \centering
    \begin{subfigure}[b]{0.3\textwidth}
        \centering
        \includegraphics[width=\textwidth]{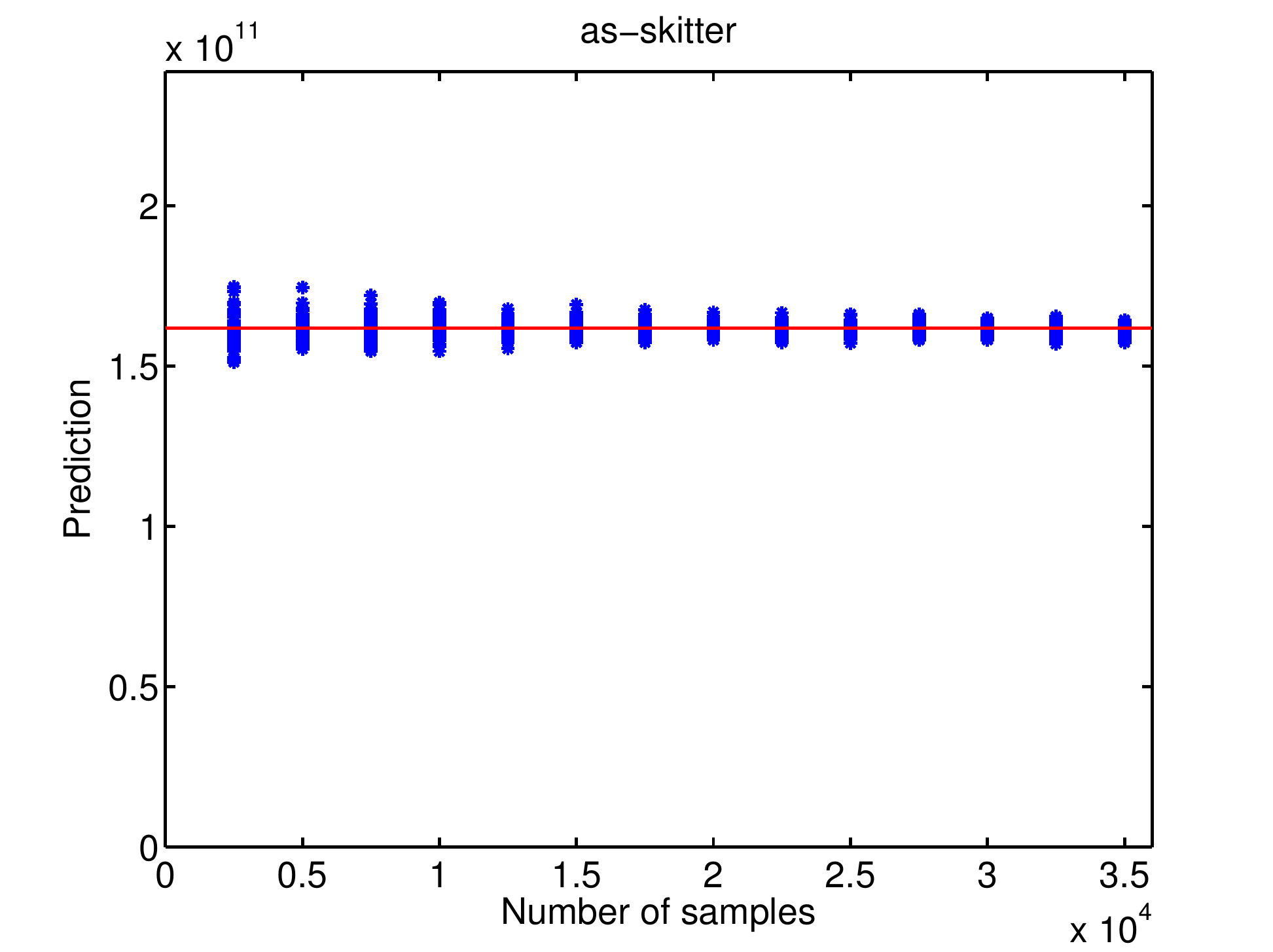}
        \caption{tailed-triangle}
        \label{fig:vp3}
    \end{subfigure}
    \hfill
    \begin{subfigure}[b]{0.3\textwidth}
        \centering
        \includegraphics[width=\textwidth]{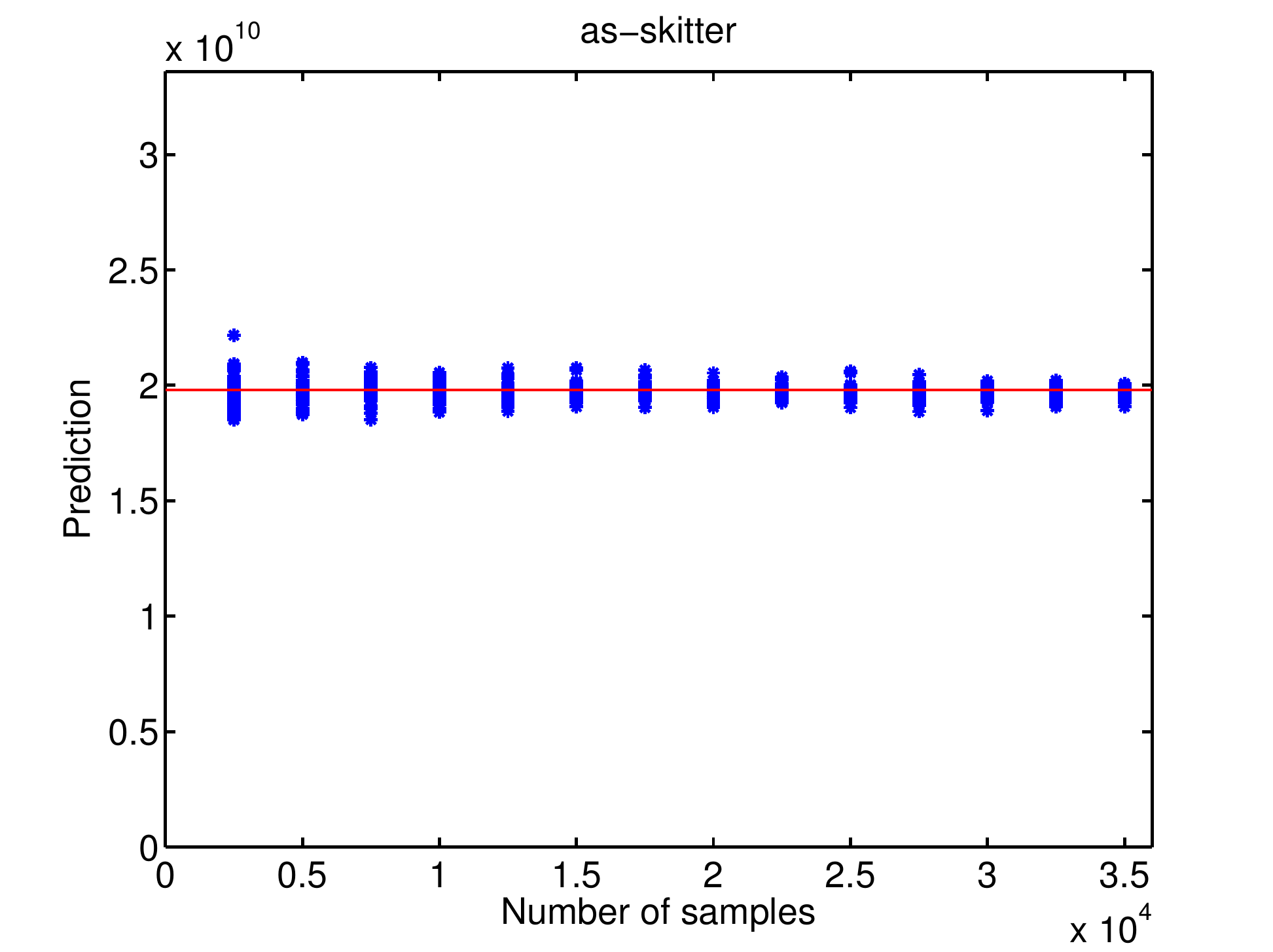}
        \caption{\footnotesize{chordal-4-cycle}}
        \label{fig:vp5}
    \end{subfigure}
    \hfill
    \begin{subfigure}[b]{0.3\textwidth}
        \centering
        \includegraphics[width=\textwidth]{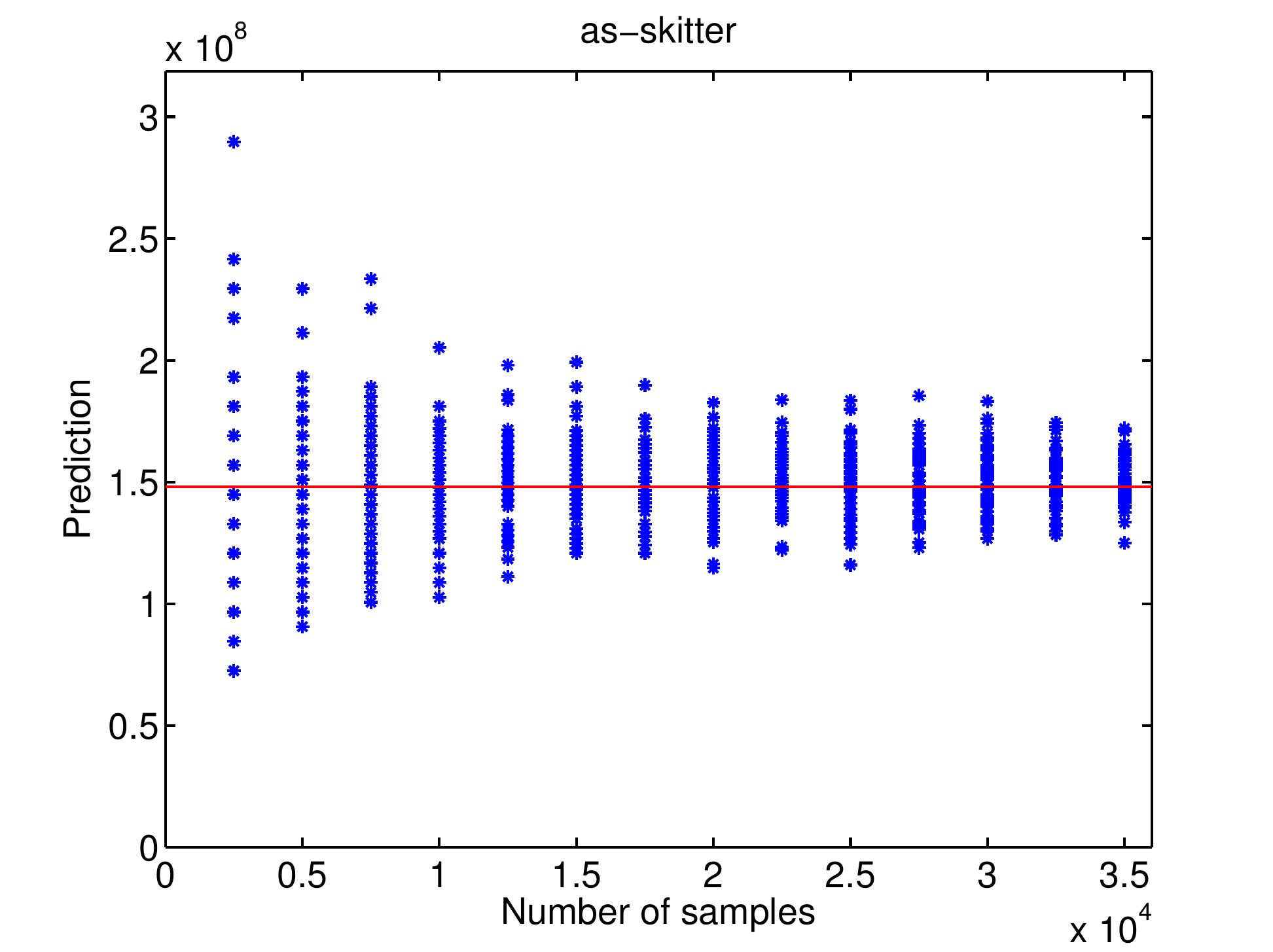}
        \caption{4-clique}
        \label{fig:vp6}
    \end{subfigure}
    \caption{Increasing number of samples decreases error: Each blue dot is an output of a run of the algorithm
    with the number of samples in the $x$-axis. The red line is the true count.}
    \label{fig:convergence}
\end{figure*}

\textbf{Speedup:}  \Fig{speedup} presented the speedups achieved over full enumeration  by using our path sampling algorithm. 
Enumeration for flickr and orkut takes order of a day. Since the motifs counts are in the order of tens of billions,
there is no hope of getting any scalability. Our algorithms takes less than a minute (even including I/O) for
all these graphs.

\textbf{The benefit of centered 3-paths:} We could simply use the basic 3-path sampling given in \estimate{} to approximate all counts.
We compare this approach to our final algorithm that use \estcent{} for $C_4, C_5, C_6$ estimates. Comparisons
between the relative errors for $C_4, C_5, C_6$
are given in \Fig{canvsvan}. (``Basic" denotes simply using \estimate, and ``centered" is the main algorithm.)
We used 200K samples for both algorithms. Some instances of using \estimate{} give somewhat large errors,
and \estcent{} really cuts these errors down. It shows the power of centered 3-path sampling.

\begin{figure*}[t]
    \centering
    \begin{subfigure}[b]{0.3\textwidth}
        \centering
        \includegraphics[width=\textwidth]{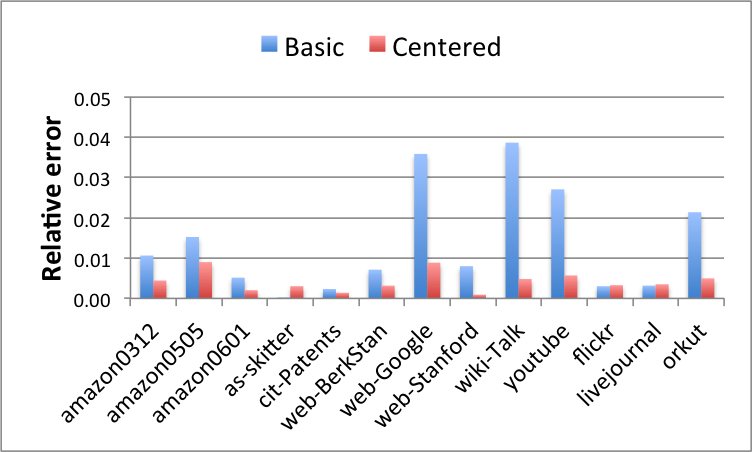}
        \caption{4-cycle}
        \label{fig:CvsVp4}
    \end{subfigure}
    \hfill
    \begin{subfigure}[b]{.3\textwidth}
        \centering
        \includegraphics[width=\textwidth]{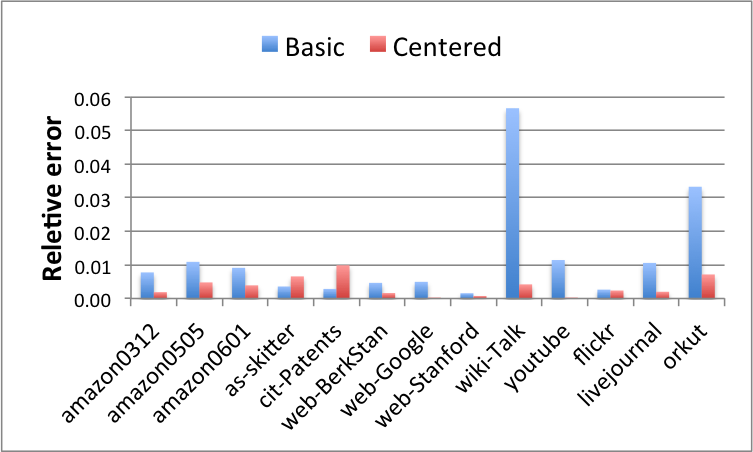}
        \caption{Chordal 4-cycle}
        \label{fig:CvsVp5}
    \end{subfigure}
    \hfill
    \begin{subfigure}[b]{0.3\textwidth}
        \centering
        \includegraphics[width=\textwidth]{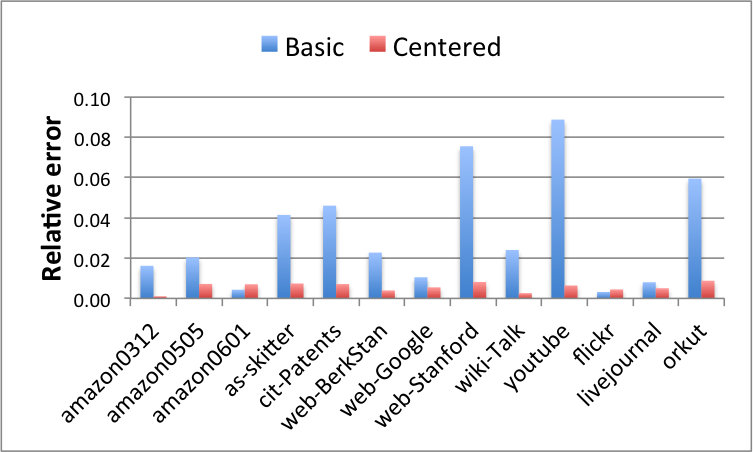}
        \caption{4-clique}
        \label{fig:CvsVp6}
    \end{subfigure}
    \caption{Comparing the accuracy of estimations using basic 3-path sampling and centered 3-path sampling}
        \label{fig:canvsvan}
\end{figure*}

\textbf{Error bounds:} We use  \Cor{cher} (as explained in \Sec{error}) to compute 99\% confidence error bounds
for all of our runs. So, for a single run of our algorithm on a candidate graph, we have a mathematical bound on the error that is solely based
on output estimates. (These are critical in the situation where we do not know the true answer, and need confidence that the estimates
are accurate.)
\Fig{errorbound} shows the accuracy of our error bounds with 99\% confidence, so $\delta = 0.01$ in \Cor{cher}.  
In all cases, the provable bounds on the error are always less than 10\% and mostly at most 5\%. (We stress that the actual error is much smaller.)
To the best of our knowledge, no previous sampling based algorithm for motif counting comes with hard mathematical error bars that are practically
reasonable.

\begin{figure}[t]
  \centering
  \includegraphics[width=.45\textwidth]{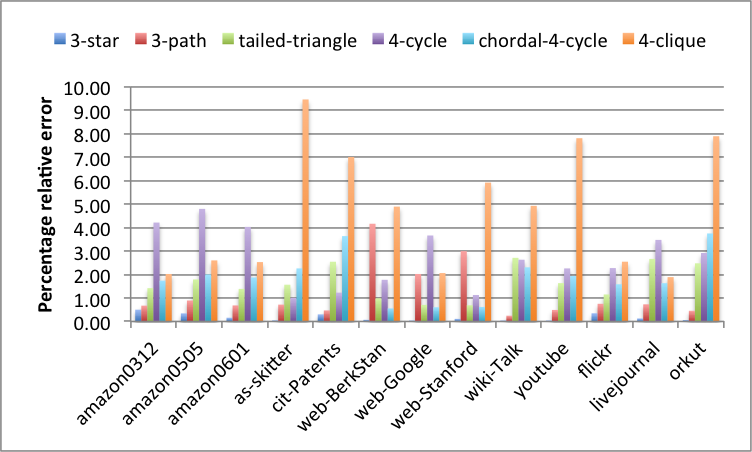}
      \caption{Provable error bounds}
      \label{fig:errorbound}
\end{figure}

\textbf{Trends in patterns:} The most frequent connected induced motif is the 
3-star. The least frequent is either the 4-cycle or the 4-clique. We find it intriguing
that (among cycle-based motifs) the chordal-4-cycle is the most frequent. (The orkut graph is a notable
exception in that 4-cycles are more frequent.) A future direction
is to connect these counts with the subgraph frequency approaches of~\cite{UganderBK13}.

\subsection{Comparison with previous work}
\label{sec:compare}

Here we present an empirical comparison between our proposed methods and other sampling based algorithms.  
We focus on color coding methods~\cite{HoBe+07,BetzlerBFKN11,ZhWaBu+12}, MCMC based sampling
algorithms~\cite{BhRaRa+12}, and edge sampling algorithms~\cite{RaBhHa14}. These methods are specifically designed for practicality and scalability. 
We give a short synopsis of these methods.

{\it GUISE \cite{BhRaRa+12}:} This employs a Markov Chain Monte Carlo (MCMC) method to uniformly sample a motif from the space of all (induced)  occurrences of motifs of sizes 3, 4, and 5. We use the implementation of GUISE from \url{http://cs.iupui.edu/~alhasan/software.html}. 

{\it  GRAFT~\cite{RaBhHa14}:} This algorithm samples a set of edges uniformly at random from the input graph and then counts the number of occurrences of each motif that uses any of these sampled edges. These counts can then be scaled appropriately to obtain unbiased estimates. 
Observe that if edges are sampled with probability 1, this gives an algorithm for exact counting using an edge iteration. 
For our experiments, we have used our own implementation. 

{\it Color Coding~\cite{AlYuZw94,HoBe+07,BetzlerBFKN11,ZhWaBu+12}:} This is a general technique for pattern counting, which samples to prune the enumeration search tree. 
We randomly color the vertices of the graph and then only enumerate over motifs all of whose vertices have distinct colors. These counts can be appropriately scaled to get unbiased estimates for the true frequencies. 
For our experiments, we implemented color coding (with $4$ colors), and used the same algorithmic techniques that we used for our enumeration algorithm. 

All these methods are quite general and work for motifs of any size. These algorithms
are more general than our approach. But our focus on the specific 
six subgraphs in \Fig{4-vertex-motifs} allows for the design of faster and more accurate algorithms,
that work better than these generic methods.

For a fair comparison, all running times for previous algorithms are for estimating frequencies of only the cycle-based motifs. 
(We are being conservative here, since we compare with the entire time for $3$-path sampling.)
Both GUISE and GRAFT take as input a number of samples. 
We ran GUISE for 10 million samples, since 
for fewer samples, the errors were usually around 100\%. At 10 million samples, the running time was 
more than the enumeration cost (usually over 5 times that), so we simply terminated.
We increased the number of samples for GRAFT
until the errors were within 5\% or the running time became larger than 5 times the enumeration cost.
For moderately sized graphs, the enumeration running time is typically
2-3 orders of magnitude more than that of $3$-path sampling. Color coding takes no parameters, and uses
$4$ colors. 

Our results are summarized in \Fig{prev-algs} (running time comparison) and \Tab{prev-algs-acc} (accuracy comparison).
For ease of presentation, we focus on 4 graphs: as-skitter, cit-patents, web-Stanford, and wiki-Talk. These range
from a million to 10 million edges. Our results were similar on other graphs. We present results for a single run
of each algorithm. (We ran for numerous iterations, and all results were consistent.)
In summary, our algorithm
is many orders of magnitude faster and more accurate than these methods.
We stress that these results are consistent with the literature, where previous methods were only
employed on graphs with about 100K edges.

\begin{figure}
    \centering
    	\includegraphics[scale=0.5]{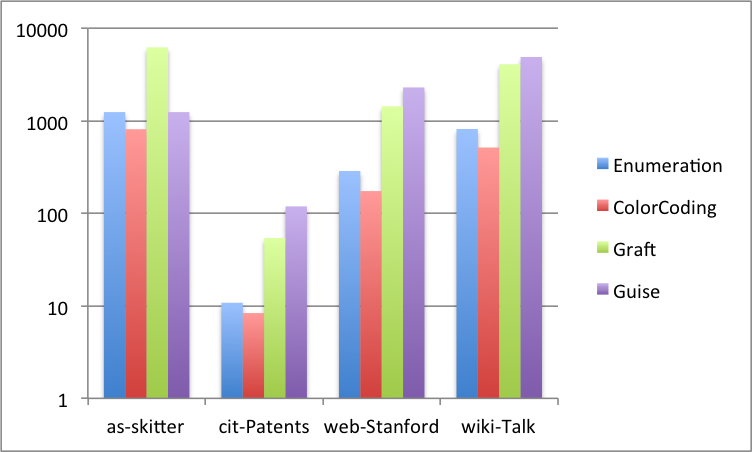}
  \caption{Runtimes of the existing algorithms normalized with respect to runtimes of 3-path sampling algorithm}
  \label{fig:prev-algs}
\end{figure}

\begin{asparaitem}
\item[\textbf{Scalability and speed.}] Previous work either employ Map-Reduce clusters
or max out at a million (or so) edges. Our $3$-path sampler runs on a single commodity machine
with 64GB memory, and easily handles graphs with more than a hundred million edges.

Running time comparisons are in \Fig{prev-algs}, where for each graph, we normalize by the running time of $3$-path sampler. 
This figure is limited to only moderate sized graphs due to large runtimes of the other algorithms.  
At about a million edges, $3$-path sampling is already hundreds of times faster than  other algorithms. 
(As we explain in the next section and \Tab{prev-algs-acc}, the accuracy of $3$-path sampling
is much better than GRAFT and GUISE, and comparable to color coding.)
We note that our enumeration code is actually competitive
with existing sampling algorithms. Color coding is consistently (a little) faster than enumeration.
GUISE and GRAFT do not yield accurate approximations even when run for many multiples of the enumeration time. We believe
that this is an issue of scale, since previous algorithm were not run for graphs with many millions of edges.

We provide some explanations for these runtimes. GRAFT provides a speedup over edge iteration enumeration procedures. 
However, edge-iteration algorithms are not the best methods  for enumeration, as is well-known for triangle enumeration~\cite{ScWa05,BeFo+14}.
Enumeration using vertex orderings is significantly faster. 
The same principles to apply to counting 4-vertex patterns. For small sizes (100K edges or less), GRAFT may
give improvements over the best enumeration, but this is not true for larger sizes. 

GUISE estimates the relative frequencies of the motifs, not the exact frequencies. GUISE has to perform relatively long random walks before it can obtain samples from the stationary distribution. This limits the number of samples that can be made. More problematically, the universe from which GUISE samples from uniformly is prohibitively large. In particular, focusing only on 4-vertex motifs, we see that (say for as-Skitter), the fraction of cliques in the universe of 4-vertex motifs is less than $10^{-6}$. This means that roughly a million samples are needed to witness a single 4-clique, 
and the square of a million samples to estimate accurately.

Color coding is hindered by the sheer size of the output. Color coding cuts down the set of motifs by enumerating only polychromatic motifs (i.e. all vertices with distinct colors), 
but this set is still quite large. In particular, for 4-vertex motifs, the reduction of the output is only of the order of one-tenth (or 3/32 to be precise -- the probability that four vertices of a motif all get different colors when using 4 colors). This means enumerating over a billion motifs in as-skitter, for example. 
There is also the extra overhead of actually searching through neighborhoods to perform this enumeration. We do see that color coding provides some benefit
over enumeration, and is probably the best algorithm  for $4$-vertex pattern counting among previous work.

\item[\textbf{Accuracy.}] We present the accuracies of the various algorithms in \Tab{prev-algs-acc}. We focus on 4-clique
counts for brevity. Even after running for times more than enumeration, GRAFT and GUISE give low accuracies. GUISE
generally fails to even find a 4-clique, and GRAFT has not processed enough samples to converge. 
Color coding is extremely accurate, and $3$-path sampling is competitive. But the running time of color coding is a hundred to thousand times
that of $3$-path sampling on these instances.

\begin{table}
\centering
\caption{Relative error in 4-clique count. For GUISE, we choose 10M samples. For GRAFT, we report the errors after running the algorithm for up to 5 times the enumeration time.}\label{tab:prev-algs-acc}
\begin{tabular}{|l|c|c|c|c|}
\hline
 & 3-path & Color Coding & Graft  & Guise\\ \hline
as-skitter &  0.7 & 0.1 & 95 &  99\\
cit-Patents &  0.7 & 0.4 & 78 &  54\\
web-Stanford &  0.8 & 0.8 & 29 & 99 \\
wiki-Talk &  0.3 & 0.5 & 5 &  99\\
\hline
\end{tabular}
\end{table}
\end{asparaitem}

\section{Conclusions and Future Work}
We get accurate results for all 4-vertex motif frequencies on a large number of graphs, and believe this is useful for motif analyses. 
Previous work usually focuses on a small, specific 
set of larger motifs~\cite{HoBe+07,BetzlerBFKN11,ZhWaBu+12}, or gives coarser approximations
for more motifs~\cite{BhRaRa+12}. 
It is natural to ask if we can extend this sampling scheme further to estimate counts of 5-vertex (or even higher order) motifs.

\end{document}